\documentclass[aps,showpacs,prl,english,10pt,a4,nofootinbib,superscriptaddress,notitlepage,twocolumn]{revtex4-1}

\usepackage{amsmath}
\usepackage{amsthm}
\usepackage{amssymb}
\usepackage{amsfonts}
\usepackage{bbm}
\usepackage{epsfig}
\usepackage{times}
\usepackage{color}
\usepackage{hyperref}
\usepackage{bm}
\usepackage{graphicx}
\usepackage{float} 
\usepackage{appendix} 

\input{Qcircuit}
\newcommand{\proj}[1]{\ket{#1}\bra{#1}}
\newcommand{\unit}{\mathbf{1}}
\newcommand{\ct}{^{\dagger}}
\newcommand{\tn}[1]{^{\otimes #1}}

\newtheorem{theorem}{Theorem}
\newtheorem{definition}[theorem]{Definition}

\begin{document} 

\title{Nonlocality in instantaneous quantum circuits}
\author{Joel J. Wallman}
\affiliation{%
	Institute for Quantum Computing and Department of Applied Mathematics,
	University of Waterloo, Waterloo, Ontario N2L 3G1, Canada
}%
\author{Emily Adlam}
\affiliation{%
	Institute for Quantum Computing and Department of Applied Mathematics,
	University of Waterloo, Waterloo, Ontario N2L 3G1, Canada
}%
\affiliation{Centre for Quantum Information and Foundations, DAMTP,
Centre for Mathematical Sciences, University of Cambridge,
Wilberforce Road, Cambridge, CB3 0WA, U.K.}


\date{\today}

\begin{abstract}
We show that families of Instantaneous Quantum Polynomial (IQP) circuits 
corresponding to nontrivial Bell tests exhibit nonlocality. However, we also 
prove that this nonlocality can only be demonstrated using post-selection or 
nonlinear processing of the measurement outcomes. Therefore if the output 
of a computation is encoded in the parity of the measurement outcomes, then 
families of IQP circuits whose full output distributions are hard to sample 
still only provide a computational advantage relative to locally causal 
theories under post-selection. Consequently, post-selection is a crucial 
technique for obtaining a computational advantage for IQP circuits (with 
respect to decision problems) and for demonstrating nonlocality within IQP 
circuits, suggesting a strong link between these phenomena.
\end{abstract}

\maketitle

It is commonly believed that quantum computers will provide an exponential speedup compared to classical computers for certain information processing tasks, such as factoring numbers~\cite{Shor1999} and simulating quantum systems~\cite{Feynman1982}. This possibility raises two natural questions: firstly, can such a computational speedup be proven rigorously, and secondly, what quantum mystery (or mysteries) is responsible for this speedup?

A number of recent results have provided rigorous evidence that quantum systems 
can indeed be exponentially hard to simulate on a classical computer under 
certain ``generic and foundational''~\cite{Aaronson2005a} assumptions about 
classical computational complexity~\cite{Terhal2004,Aaronson2005a}.

Given the evidence for such a computational advantage, it is of great practical and foundational importance to identify the exact ingredients required for a quantum computer to outperform any classical computer. A first step in this direction is to identify maximal sets of quantum operations that are efficiently simulable on a classical computer so that the resources providing a quantum advantage then residing outside such sets. The canonical result along these lines is the Gottesman-Knill theorem, which shows that circuits consisting of stabilizer operations can be efficiently simulated~\cite{Aaronson2004}. For odd-prime dimensional systems, the Gottesman-Knill theorem can be generalized to allow for all operations with nonnegative discrete Wigner function~\cite{Mari2012,Veitch2012}, and allowing any other quantum states then gives a proof of contextuality~\cite{Howard2014}, suggesting that contextuality is the key quantum phenomenon that enables a quantum computational advantage.

An alternative approach is to try and identify small sets of quantum operations 
that provably cannot be efficiently simulated classically (under reasonable 
complexity theoretic assumptions). Of particular interest in this regard is the 
family of Instantaneous Quantum Polynomial (IQP) circuits consisting of 
preparations and measurements in the Pauli-$X$ eigenbasis together with 
unitaries that are diagonal in the Pauli-$Z$ eigenbasis~\cite{Shepherd2009}. 
While the quantum operations in IQP circuits appear simple, the output of 
families of IQP circuits generally cannot be efficiently sampled (i.e., cannot 
be efficiently weakly simulated~\cite{VandenNest2011}) unless the polynomial 
hierarchy collapses to the third level~\cite{Bremner2011}. Since IQP circuits 
are a simple class of circuits which are not classically simulable, they 
provide a clear framework in which to identify the nonclassical phenomenon(a) 
responsible for a quantum computational speedup and the operational 
capabilities required to demonstrate them.

In this spirit, it has recently been shown that a class of IQP circuit families, known as \textbf{IQP}*, retain the property of being difficult to sample while exhibiting no nonlocality (when the circuit is divided into a state-preparation procedure and measurement gadgets as in Fig.~\ref{fig:Bell_test})~\cite{Hoban2014}, since they cannot be used to evaluate nonlinear boolean functions~\cite{Hoban2011b}. 

However, the definition of \textbf{IQP}* is too restrictive for nonlocality to be a relevant property, since it only allows one measurement basis at each site. In order to violate a Bell inequality, multiple measurements per site must be possible, since otherwise the local hidden variable can simply be a string of measurement outcomes sampled from the quantum probability distribution.

In this paper we consider familes of IQP circuits that correspond to nontrivial Bell tests. We show that such families do in fact exhibit nonlocality and are also hard to sample classically. However, we prove that this nonlocality can only be demonstrated using post-selection or nonlinear processing, and so such families of circuits do not provide any computational advantage over local hidden variable theories when classical side-processing is restricted to linear computations. 

\section{IQP circuits} 

We begin by defining IQP circuits and the uniformity conditions that define the associated complexity classes \textbf{IQP} and \textbf{IQP}$^*$.
\begin{definition} 
An $n$-qubit IQP circuit $C_{n,x}$ with input bit-string $x\in\mathbb{Z}_2^n$ 
consists of
\begin{enumerate} 

\item a quantum register prepared in the input state $ |+ \rangle^{\otimes n} $;

\item a unitary operator $U_{n,x}$ applied to the register, where $U_{n,x}$ is 
diagonal with respect to the Pauli-$Z$ eigenbasis operators; and

\item a Pauli-X basis measurement on a subset of the qubits. 

\end{enumerate} 

\end{definition} 

IQP circuits are so-called because $U_{n,x}$ may be decomposed into a product of diagonal gates, which commute and therefore may be applied in any order, or, in particular, simultaneously.

In order to ensure that additional computational power is not hidden in the 
circuit descriptions, it is important to enforce a \emph{uniformity condition} 
ensuring that descriptions of circuits can be efficiently generated 
classically. The original (and most permissive) uniformity condition proposed 
for a family $\{C_{N,x}:n\in\mathbb{N},x\in\mathbb{Z}_2^n\}$ of IQP to be in 
the complexity class \textbf{IQP} is that the unitaries $U_{n,x}$ can all be 
efficiently (with respect to $n$) described by a classical computer. The 
complexity class \textbf{IQP}$^*$ was defined in Ref.~\cite{Hoban2014} by 
further requiring that $U_{n,x}=U_n Z[x]$ for some diagonal $U_n$, where $Z[x] 
= \otimes_{m=1}^n Z^{x_m}$. An equivalent statement of the \textbf{IQP}$^*$ 
uniformity condition is that the input state is $\bigotimes_{m=1}^n 
\ket{-^{x_m}}$ and that $U_n$ is independent of $x$, which is directly 
analogous to the uniformity condition for $\textbf{BPP}$ (i.e., universal 
classical computation).

The fact that the outputs of circuit families in \textbf{IQP} and \textbf{IQP}$^*$ are hard to sample classically rests upon the Hadamard post-selection gadget depicted in Fig.~\ref{fig:HadamardGadget}, since the $\tfrac{\pi}{8}$-gate $R(\tfrac{\pi}{4})$, where $R(\theta) = \proj{0} + e^{i\theta}\proj{1}$, and the controlled-$Z$ gate $\Delta(Z) = \proj{0}\otimes\unit + \proj{1}\otimes Z$ are diagonal gates (and so can be implemented in IQP circuits) and $\{H,R(\tfrac{\pi}{4}),CZ\}$ is a universal gate set.

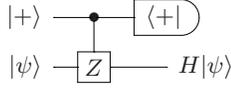
\begin{figure}[!t]
	\begin{tabular}{c c c}
		\Qcircuit @C=1em @R=.7em {\lstick{\ket{+}}
		& \ctrl{1} &  \measureD{\bra{+}} \\  \lstick{\ket{\psi}}
		& \gate{Z} & \rstick{ H\ket{\psi }} \qw
	}
	\end{tabular}
	\caption{The Hadamard post-selection gadget. Post-selected on the outcome of the Pauli-$X$ measurement on the top line being $\bra{+}$, the state $\ket{\psi}$ is mapped to $H\ket{\psi}$, where $H$ is the Hadamard gate.}
	\label{fig:HadamardGadget}
\end{figure}

It was shown in Ref.~\cite{Hoban2014} that a set of circuits $\{C_{n,x}:x\in\mathbb{Z}_2^n\}$ for fixed $n\in\mathbb{N}$ satisfying the uniformity condition for \textbf{IQP}$^*$ could never violate a Bell inequality. The proof of this explicitly rests on the fact that if $U_{n,x} = U_n Z[x]$, then the $Z$ gate can be commuted to the end of the circuit and then absorbed into the measurement, which can be simulated by post-processing the measurement outcomes on each qubit individually (i.e., locally).

However, this trivial dependence on the input, while well-motivated by 
classical analogy, precludes the possibility of any nontrivial dependence of 
the computation on the input, rendering computations in \textbf{IQP}$^*$ 
automatically consistent with a local hidden variable theory. To resolve this 
problem, we now define an IQP Bell test in a way that allows violations of Bell 
inequalities. The following definition is motivated by the depiction of a Bell 
test and its equivalent circuit representation in Fig.~\ref{fig:Bell_test}.

\begin{definition} 
	An IQP Bell test is a family of IQP circuits 
	$\{C_{n,x}:x\in\mathbb{Z}_2^n\}$ for a fixed $n\in\mathbb{N}$ such that 
	$U_{n,x} = U_n \bigotimes_{m=1}^n D_m^{x_m}$, where $U_n$ is a diagonal 
	gate and $\{D_m:m=1,\ldots,n\}$ are single-qubit diagonal gates.
\end{definition}

We now prove that IQP Bell tests can be nontrivial Bell tests, that is, that 
there exist IQP Bell tests that are inconsistent with any local hidden-variable 
theory.

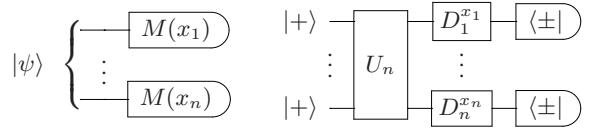
\begin{figure}[!t]
	\begin{minipage}[b]{0.45\linewidth}
	\begin{tabular}{c c c c}
		\Qcircuit @C=1em @R=.7em {
		& &\qw &  \measureD{M(x_1)} \\
			\lstick{\ket{\psi}} & &\vdots & \\
			& &\qw & \measureD{M(x_n)} \gategroup{1}{2}{3}{2}{0.7em}{\{}
		}
	\end{tabular}
	\end{minipage}
	\hspace{1mm}
	\begin{minipage}[b]{0.45\linewidth}
		\begin{tabular}{c c c c}
			\Qcircuit @C=1em @R=.7em {
				\lstick{\ket{+}} & \multigate{3}{U_n} &\gate{D_1^{x_1}}&  
				\measureD{\bra{\pm}} \\
				\vdots &  & \vdots & \\
				 &  & & \\
			\lstick{\ket{+}} & \ghost{U_n}	& \gate{D_n^{x_n}} & 
			\measureD{\bra{\pm}}
			}
		\end{tabular}
	\end{minipage}
	\caption{(left) Depiction of an $n$-qubit Bell state, wherein an $n$-qubit 
	state $\ket{\psi}$ is distributed amongst $n$ space-like separated parties. 
	The $m$th party performs one of two measurements parametrized by $x_m$, and 
	returns the outcome $z_m$ to the verifier. Iterating the experiment many 
	times, the verifier obtains joint probabilities $p(z|x)$ and can then 
	determine whether the statistics are consistent with a local 
	hidden-variable theory by testing whether the statistics violate a Bell 
	inequality. (right) An equivalent reformulation of the Bell scenario into a 
	set of circuits, where $U_n$ acting on $\ket{+}\tn{n}$ can be regarded as a 
	state-preparation gadget and the measurement bases are 
	$\{D_m^{-x_m}\ket{\pm}\}$.}
	\label{fig:Bell_test}
\end{figure}

\begin{theorem}\label{thm:IQP-Bell} 
	There exist IQP Bell tests that are inconsistent with any local hidden-variable theory.
\end{theorem}

\begin{proof}
Consider the circuit depicted in Fig.~\ref{fig:GHZgadget}, which can be regarded as a post-selection gadget for preparing GHZ states. After all the gates are applied, the five qubits are in the state
\begin{align}\label{eq:distributed_state}
\ket{\psi}&=\ket{++}_{12}(\ket{000}_{345} - \ket{111}_{345})\notag\\
&\quad + \ket{+-}_{12}(\ket{001}_{345} - \ket{110}_{345})\notag\\
&\quad + \ket{-+}_{12}(\ket{100}_{345} - \ket{011}_{345})\notag\\
&\quad + \ket{--}_{12}(\ket{010}_{345} - \ket{101}_{345}).
\end{align}

We can then distribute the five qubits amongst five space-like separated parties. The first two parties measure in the Pauli-$X$ basis, while the other parties measure in the Pauli-$X$ basis if $x_m = 0$ and the Pauli-$Y$ basis if $x_m = 1$. The measurement outcomes $\{z_m:m=3,\ldots,5\}$ can readily be shown to satisfy~\cite{Anders2009}
\begin{align}
	z= z_3 \oplus z_4 \oplus z_5 = \delta_{x_3 x_4\oplus 1}
\end{align}
when post-selected on $z_1=z_2 = 0$ and $x_5 = x_3 \oplus x_4$, where the outcome $0$ corresponds to $\ket{+}$ or $\ket{i}$ respectively and $\oplus$ denotes addition modulo two.

We now show that
\begin{align}\label{eq:GHZ_Bell}
	\sum_{x_3,x_4} {\rm Pr}(z = x_3 x_4 \oplus 1) \leq 3
\end{align}
is a Bell inequality under this post-selection. The proof is a simple 
modification of those in Ref.~\cite{Hoban2011c}.

In a local hidden variable theory, preparing the state $\ket{\psi}$ in Eq.~\eqref{eq:distributed_state} corresponds to preparing a hidden variable $\lambda\in\Lambda$, which is then (without loss of generality) sent to each party. Without loss of generality, all measurement outcomes can be assumed to be deterministic~\cite{Fine1982} and the outcome of the $m$th party's measurement must be independent of both the measurement setting $x_{m'}$ and outcome $z_{m'}$ for all $m'\neq m$, so the measurement outcomes for $m=3,4,5$ can be written as
\begin{align}
 z_m(\lambda,x|z_1=0,z_2=0) = a_m(\lambda)x_m \oplus b_m(\lambda)
\end{align}
for some functions $a_m,b_m:\Lambda\to \mathbb{Z}_2$.

Consequently, even when post-selecting on $z_1=z_2=0$, the parity of the remaining outcomes is 
\begin{align}
 z(\lambda,x)&= z_3(\lambda,x) \oplus z_4(\lambda,x) \oplus z_5(\lambda,x)\notag\\
 &= a_3(\lambda)x_3 \oplus a_4(\lambda)x_4 \oplus a_5(\lambda)x_5 \oplus b'(\lambda)	\,.
\end{align}
Since $x_3,x_4,x_5$ are independent of $\lambda$ (which is often referred to as measurement independence or, more controversially, as a ``free will'' assumption), post-selecting on $x_5 = x_3\oplus x_4$ gives
\begin{align}
z(\lambda,x) = f_{\lambda}(x_3,x_4)	\,,
\end{align}
where $f_{\lambda}:\mathbb{Z}_2^2\to\mathbb{Z}_2$ is a linear function. We then have
\begin{align}
\sum_{x_3,x_4} {\rm Pr}(z = x_3 x_4 \oplus 1) &= \sum_{x_3,x_4}\int_{\Lambda}\rm{d}\lambda\,\delta_{f_{\lambda}(x_3,x_4)\oplus x_3 x_4 \oplus 1}\notag\\
&\leq \max_f \sum_{x_3,x_4}\delta_{f_{\lambda}(x_3,x_4)\oplus x_3 x_4 \oplus 1}\notag\\
&= 3\,,
\end{align}
where the maximization in the third line is over linear functions $f:\mathbb{Z}_2^2\to\mathbb{Z}_2$ and the measure $\rm{d}\lambda$ is the measure for the preparation procedure under the specified post-selections.
\end{proof}

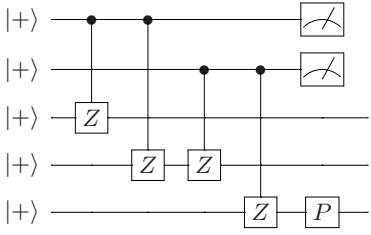
\begin{figure}[!t]
	\begin{tabular}{c c c c c c c}
	\Qcircuit @C=1em @R=.7em {\lstick{\ket{+}}
		& \ctrl{2} & \ctrl{3} & \qw & \qw &  \meter & \\ 
		\lstick{\ket{+}} & \qw & \qw & \ctrl{2} & \ctrl{3} &  \meter & \\ 
		\lstick{\ket{+}} & \gate{Z} & \qw & \qw & \qw & \qw & \qw\\ 
		\lstick{\ket{+}} & \qw & \gate{Z} & \gate{Z} & \qw & \qw & \qw\\ 
		\lstick{\ket{+}} & \qw & \qw & \qw & \gate{Z} & \gate{P} & \qw
	}
	\end{tabular}
	\caption{GHZ preparation gadget. Any measurement outcomes of Pauli-$X$ measurement on the first two lines will produce a maximally entangled state.}
	\label{fig:GHZgadget}
\end{figure}

While there are IQP Bell tests inconsistent with any local hidden variable theory, the above Bell test requires post-selection. We now show that this is a generic feature of any proof of nonlocality for IQP Bell tests. More precisely, we will show that any linear function of the measurement outcomes in an IQP Bell test is consistent with a local hidden variable theory. Some form of nonlinearity, either by explicitly testing a nonlinear function of the measurement outcomes or by post-selecting on desired measurement outcomes, is therefore required to reveal nonlocality in IQP Bell tests. That is, IQP Bell tests exhibit ``hidden nonlocality'' of the form considered elsewhere~\cite{Popescu1995,Hirsch2013}, with the caveat that the nonlocality is only hidden with respect to the allowed processing of measurement outcomes.

To see that post-selection can be viewed as a form of nonlinearity, note that in the example in the proof of the above theorem, post-selecting on the first two outcomes being $z_1=z_2 = 0$ and violating Eq.~\eqref{eq:GHZ_Bell} is equivalent to
\begin{align}
	\sum_{z_1,z_2,x_3,x_4} {\rm Pr}(\bar{z}_1\bar{z}_2[z \oplus x_3 x_4 \oplus 1]=0|z_1,z_2) \leq 11	\,, 
\end{align}
where $\bar{a} = a\oplus 1$ denotes the classical not, since
\begin{align}
\sum_{x_3,x_4} {\rm Pr}(\bar{z}_1\bar{z}_2[z \oplus x_3 x_4 \oplus 1]=0|z_1,z_2) =4
\end{align}
whenever $\bar{z}_1\bar{z}_2=0$.

\begin{theorem}\label{thm:require_postselection}
	Any linear function of the measurement outcomes of an IQP Bell test is consistent with a local hidden variable theory.
\end{theorem}

\begin{proof}

Any linear function of the measurement outcomes of an $n$-qubit IQP Bell test can be written as the parity of a subset of $k$ measurement outcomes, which, without loss of generality, we set to be the last $k$ parties. The quantum correlation functions are
\begin{align}\label{eq:correlations}
 E(x) &= {\rm Pr}(z = 0|x) - {\rm Pr}(z = 1|x) \notag\\
 &=  \bra{+}^{\otimes n}U_n^{\dagger} D(x)^{\dagger} X[1_{n,k}] D(x) U_n \ket{+}^{\otimes n}	\,,
\end{align}
where $z = \bigoplus_{m=n-k+1}^{n}z_n$, $D(x) = \bigotimes_{m=1}^n D_m^{x_m}$, $1_{n,k}$ is the $n$-bit string whose last $k$ entries are 1 and the remainder are zero and $A[s]=\bigotimes_{m=1}^n A^{s_m}$.

First note that if $k< n$, we can write $U_n = \ket{0}\bra{0}\otimes V + 
\ket{1}\bra{1}\otimes W$ to obtain
\begin{align}
E(x) &= \frac{1}{2}\bra{+}^{\otimes n'}V^{\dagger} D(x)^{\dagger} X[1_{n'}] D(x) V \ket{+}^{\otimes n'}	\notag\\
&\quad + \frac{1}{2}\bra{+}^{\otimes n'}W^{\dagger} D(x)^{\dagger} X[1_{n'}] 
D(x) W \ket{+}^{\otimes n'}
\end{align}
where $n'=n-1$, that is, the correlation functions are convex combinations of correlation functions for $n'$-qubit IQP Bell tests. Iterating as necessary, we see that the correlation functions are equivalent to convex combinations of full correlation functions for $k$-qubit IQP Bell tests (that is, to correlation functions for the parity of the full set of $k$ outcomes). 

Therefore if the full correlation functions (i.e., when $n=k$) for all $n$-qubit IQP Bell tests are consistent with locally causal theories, then all linear functions of the measurement outcomes of $n$-qubit IQP Bell tests are consistent with a locally causal theory.

We now show that the full correlation functions for $n$-qubit IQP Bell tests take the simple form
\begin{align}\label{eq:simple_correlators}
E(x) = \sum p_{f,c}\cos(f\cdot x + c)
\end{align}
for some vectors $f\in\mathbb{R}^n$ and constant $c\in\mathbb{R}$. We can 
write any $n$-qubit diagonal gates $R(x)$ and $U(x)$ (where the $x$-dependence 
is as in the definition of an IQP Bell test) as 
\begin{align}
R(x) &= \ket{0}\bra{0}\otimes S(x') + e^{i g x_1}\ket{1}\bra{1}\otimes T(x') \notag\\
U(x) &= \ket{0}\bra{0}\otimes V(x') + e^{i h x_1}\ket{1}\bra{1}\otimes W(x')
\end{align}
for some $n'$-qubit unitaries $S(x')$, $T(x')$, $V(x')$ and $W(x')$ and real constants $g$ and $h$, where $n'=n-1$ as before and $x'$ is the $n'$-bit string obtained by deleting the first entry of $x$.

We then have
\begin{widetext}
\begin{align}
\bra{+}^{\otimes n} R(x)\ct X^{\otimes n}U(x)\ket{+}^{\otimes n} &= \frac{1}{2}e^{i h_1 x_1} \bra{+}^{\otimes n'} S(x')\ct X^{\otimes n'}W(x')\ket{+}^{\otimes n'} +\frac{1}{2}e^{-i g_1 x_1} \bra{+}^{\otimes n'} T(x')\ct X^{\otimes n'}V(x')\ket{+}^{\otimes n'}
\end{align}
\end{widetext}
The correlation functions defined in Eq.~\eqref{eq:correlations} are a special 
case of the above where $U(x)=R(x)$. Therefore iteratively applying the above 
argument and noting that all phases are linear functions of $x$ and appear in 
conjugate pairs (with no relative phase) gives 
Eq.~\eqref{eq:simple_correlators}.

Finally, the full correlation functions for an $n$-qubit IQP Bell test are consistent with a locally causal theory if and only if they satisfy the Werner-Wolf-\.Zukowski-Brukner (WWZB) inequality~\cite{Werner2001,Zukowski2002}
\begin{align}\label{eq:WWZB}
\mathcal{S}_{\rm WWZB}=\sum_{a\in\mathbb{Z}_2^n}\lvert \sum_{x\in\mathbb{Z}_2^n} (-1)^{a\cdot x}E(x) \rvert \leq 2^n.
\end{align}
Since this inequality is convex in $E(x)$ (as can be immediately seen by appealing to the triangle inequality), we need only prove that any correlation functions of the form
\begin{align}\label{eq:convex_correlation}
E(x)=\cos(f\cdot x + c) = \frac{1}{2}(e^{i f\cdot x  + i c} + e^{-i f\cdot x  - i c}) 
\end{align}
for arbitrary $f\in\mathbb{R}^n$ and $c\in\mathbb{R}$ will satisfy the WWZB 
inequality. Using trivial algebraic manipulations and the identity
\begin{align} 
\cos(a_m\pi/2 -f_m/2) = (-1)^{a_m}\cos(a_m\pi/2 + f_m/2)
\end{align}
for $a_m=0,1$, we find that for any fixed $a$ and correlation functions as in 
Eq.~\eqref{eq:convex_correlation},
\begin{align}
\sum_{x\in\mathbb{Z}_2^n} (-1)^{a\cdot x} E(x)&= 2^n \cos\left(c + \sum_m 
\frac{a_m\pi + f_m}{2}\right)\notag\\
&\quad \times\prod_{m=1}^n \cos\left(\frac{a_m\pi + f_m}{2}\right).
\end{align}

Therefore for any pairs of terms $a$ and $a'$ in Eq.~\eqref{eq:WWZB} such 
that $a_1 =0$, $a'_1 = 1$ and all other elements of $a$ and $a'$ coincide,
\begin{widetext}
\begin{align}
\lvert \sum_{x\in\mathbb{Z}_2^n} (-1)^{a\cdot x}E(x) \rvert + \lvert 
\sum_{x\in\mathbb{Z}_2^n} (-1)^{a'\cdot x}E(x) \rvert &= 2^n \left[\vert 
\cos\theta\cos\left(f_1/2\right) \vert + \vert 
\sin\theta\sin\left(f_1/2\right) \vert\right] \prod_{m=2}^n 
\vert\cos\left(\frac{a_m\pi + f_m}{2}\right)\vert	\notag\\
&= 2^n \vert 
\cos\left(\theta \pm f_1/2\right) \vert \prod_{m=2}^n 
\vert\cos\left(\frac{a_m\pi + f_m}{2}\right)\vert
\end{align}
\end{widetext}
where $\theta=c+f_1/2 +\sum_{m=2}^n (a_m\pi + f_m)/2$ and the relative sign is 
positive if the signs of $\cos\theta\cos\left(f_1/2\right)$ and 
$\sin\theta\sin\left(f_1/2\right)$ are different and negative if the signs are 
the same. Iterating this argument gives
\begin{align}
\mathcal{S}_{\rm WWZB} \leq 2^n
\end{align}
for any IQP Bell test, so the correlation functions are consistent with a 
locally causal theory.
\end{proof}

\section{Discussion} 

We have shown that families of IQP circuits do in fact exhibit nonlocality, 
and, consequently, the possibility of a quantum speedup using such circuits can 
be attributed to nonlocality. Crucially, this nonlocality is only revealed by 
the very operation used to prove that IQP circuits are hard to sample, namely 
post-selection~\cite{Bremner2011}, establishing a concrete connection between 
nonlocality and a computational advantage.

These results help us understand why IQP circuits, despite their apparent 
simplicity, cannot be efficiently simulated classically (unless the polynomial 
hierarchy collapses): the probability distributions produced by IQP circuits 
possess nonlocality that is hidden with respect to linear classical 
side-processing, analogous to ``hidden 
nonlocality''~\cite{Popescu1995,Hirsch2013}. This hidden nonlocality shows that 
the output distributions of IQP circuits possess a level of additional 
structure that is not evident until the results are analyzed in a nontrivial 
manner and yet must be reproduced in a classical simulation, making sampling 
IQP circuits more difficult than we might otherwise expect.

While our results have been derived by considering prepare-and-measure 
scenarios, they also have implications for dynamics. A model is dynamically 
local if there is a local hidden variable for each qubit, unitaries correspond 
to local stochastic maps, and measurement outcomes only depend on the hidden 
variable of the qubit being measured~\cite{pusey_masters}. Our results then 
imply that there is a form of ``dynamical nonlocality'' in IQP circuits, 
similar to that exhibited in the Aharanov-Bohm 
effect~\cite{Aharonov1959,Popescu2010}.

However, it is possible that a dynamically local model could correctly 
reproduce the statistics for a family of IQP circuits (while necessarily 
failing for other circuits). Such a dynamically local model for a family of 
circuits automatically provides an efficient simulation algorithm for the 
family of circuits, provided the model can be constructed and verified 
efficiently. Therefore a pressing open problem is to determine whether there 
exist no dynamically local models for specific families of IQP circuits that are
hard to simulate, or whether such models are simply hard to construct and 
verify.

\textit{Acknowledgements.} J.~J.~W. acknowledges helpful discussions with 
M.~Bremner, J.~Emerson, M.~Hoban, M.~Howard, R.~Raussendorf, E.~Rieffel, and 
H.~Wiseman. This research was supported by CIFAR, the Government of 
Ontario, and the Government of Canada through NSERC and Industry Canada.

\end{document}